\documentclass[showpacs,amsmath,amssymb,onecolumn,pra]{revtex4-1}

\usepackage[dvips]{graphicx} 
\usepackage{amsfonts}
\usepackage{amssymb}
\usepackage{amscd}
\usepackage{amsmath}    
\usepackage{enumerate}
\usepackage{epsfig}
\usepackage{subfigure}
\usepackage{xcolor}
\usepackage{amsthm}
\usepackage{framed}
\usepackage{multirow}
\usepackage{color}

\usepackage{hyperref}

\newtheorem{theorem}{Theorem}
\newtheorem{lemma}{Lemma}
\newtheorem{corollary}{Corollary}

\newtheorem{definition}{Definition}

\newcommand{\bra}[1]{\mbox{$\left\langle #1 \right|$}}
\newcommand{\ket}[1]{\mbox{$\left| #1 \right\rangle$}}
\newcommand{\braket}[2]{\mbox{$\left\langle #1 | #2 \right\rangle$}}
\newcommand{\comments}[1]{}
\begin{document}
\preprint{APS/123-QED}
\title{Polynomial measure of coherence}
\date{\today}
\author{You Zhou}
\author{Qi Zhao}
\author{Xiao Yuan}
\author{Xiongfeng Ma}
\affiliation{Center for Quantum Information, Institute for Interdisciplinary Information Sciences, Tsinghua University, Beijing 100084, China}

\begin{abstract}
Coherence, the superposition of orthogonal quantum states, is indispensable in various quantum processes. Inspired by the polynomial invariant for classifying and quantifying entanglement, we first define polynomial coherence measure and systematically investigate its properties. Except for the qubit case, we show that there is no polynomial coherence measure satisfying the criterion that its value takes zero if and only if for incoherent states. Then, we release this strict criterion and obtain a necessary condition for polynomial coherence measure. Furthermore, we give a typical example of polynomial coherence measure for pure states and extend it to mixed states via a convex-roof construction. Analytical formula of our convex-roof polynomial coherence measure is obtained for symmetric states which are invariant under arbitrary basis permutation. Consequently, for general mixed states, we give a lower bound of our coherence measure.
\end{abstract}


\maketitle

\section{Introduction}
Coherence describes a unique feature of quantum mechanics --- superposition of orthogonal states. The study of coherence can date back to the early development of quantum optics \cite{glauber1963quantum}, where interference phenomenon is demonstrated for the wave-particle duality of quantum mechanics. In quantum information, coherence acts as an indispensable ingredient in many tasks, such as quantum computing \cite{nielsen2010quantum}, metrology \cite{BRAUNSTEIN1996135}, and randomness generation \cite{Ma2016QRNG}. Furthermore, coherence also plays an important role in quantum thermodynamics \cite{aaberg2014catalytic,lostaglio2015description,lostaglio2015quantum}, and quantum phase transition \cite{ccakmak2015factorization,karpat2014quantum}.

With the development of the quantum information theory, a resource framework of coherence has been recently proposed \cite{Baumgratz2014Quantify}. The free state and the free operation are two elementary ingredients in a quantum resource theory. In the resource theory of coherence, the set of free states is a collection of all quantum states whose density matrices are diagonal in a reference computational basis ${I} = \{\ket{i}\}$. The free operations are incoherent complete positive and trace preserving (ICPTP) operations, which cannot map any incoherent state to a coherent state. With the definitions of free states and free operations, one can define a coherence measure that quantifies the superposition of reference basis. Based on this coherence framework, several measures are proposed, such as relative entropy of coherence, $l_1$ norm of coherence \cite{Baumgratz2014Quantify}, and coherence of formation \cite{yuan2015intrinsic,winter2016operational}. Moreover, coherence in distributed systems \cite{Chitambar2016Relating,streltsov2015hierarchies} and the connections between coherence and other quantum resources are also developed along this line \cite{ma2016converting,streltsov2015measuring,chitambar2016assisted}.

One important class of coherence measures is based on the convex-roof construction \cite{yuan2015intrinsic}. For any coherence measure of pure states $C(\ket{\psi})$, the convex roof extension of a general mixed state $\rho$ is defined as
\begin{equation}\label{convex roof}
\begin{aligned}
C(\rho)=\min\limits_{\{p_i,\ket{\psi_i}\}} \sum_i p_i C(\ket{\psi_i}),
\end{aligned}
\end{equation}
where the minimization is over all the decompositions $\{p_i,\ket{\psi_i}\}$ of  $\rho=\sum_i p_i \ket{\psi_i}\bra{\psi_i}$. When $C(\ket{\psi})=S(\Delta(\ket{\psi}\bra{\psi}))$, where $S$ is von Neumann entropy and $\Delta(\rho) = \sum_i \ket{i}\bra{i}\rho\ket{i}\bra{i}$ is the dephasing channel on the basis $I$, the corresponding measure is the coherence of formation. When $C(\ket{\psi})=\max_i |\braket{i}{\psi}|^2$, the corresponding measure is the geometric coherence \cite{streltsov2015measuring}. In general, the minimization problem in Eq.~\eqref{convex roof} is extremely hard. In particular, analytical formula of the coherence of formation is only obtained for qubit states.

This is very similar to quantifying another well-known quantum resource, entanglement, where free states are separable states and free operations are local operations and classical communication \cite{horodecki2009quantum}. In entanglement measures, convex-roof constructions have been widely studied \cite{bennett1996mixed,uhlmann1998entropy}. Similarly, the minimization problem is generally hard. Fortunately, there are two solvable cases, concurrence \cite{hill1997entanglement,wootters1998entanglement} and three-tangle \cite{coffman2000distributed}. Both of them are related to a very useful class of functions, referred as \emph{polynomial invariant} \cite{eltschka2014quantifying}. A polynomial invariant is a homogenous polynomial function of the coefficients of a pure state, $P_h(\ket{\psi})$, which is invariant under stochastic local operations and classical communication (SLOCC) \cite{dur2000three}. Denote $h$ to be the degree of the polynomial function, for an $N$-qudit state $\ket{\psi}$,
\begin{equation} \label{}
\begin{aligned}
P_h(\kappa L\ket{\psi})=\kappa^hP_h(\ket{\psi}),
\end{aligned}
\end{equation}
where $\kappa$ is an arbitrary scalar and $L\in \mathcal{SL}(d,\mathbb{C})^{\otimes{N}}$ is a product of invertible linear operators representing SLOCC. For an entanglement measure of pure states, one can add a positive power $m$ to the absolute value of the polynomial invariant,
\begin{equation}\label{}
\begin{aligned}
E_h^m(\ket{\psi})=|P_h(\ket{\psi})|^m,
\end{aligned}
\end{equation}
where the overall degree is $hm$. Polynomial invariants are used to classify and quantify various types of entanglement in multi-qubit \cite{osterloh2005constructing,d2008polynomial} and qudit systems \cite{gour2013classification}. Specifically, the convex-roof of concurrence can be solved analytically in the two-qubit case \cite{wootters1998entanglement}, and the three-tangle for three-qubit is analytically solvable for some special mixed states \cite{lohmayer2006entangled,jung2009three,siewert2012quantifying}. Recently, a geometric approach \cite{regula2016entanglement} is proposed to analyse the convex-roof extension of polynomial measures for the states of more qubits in some specific cases.

Inspired by the polynomial invariant in entanglement measure, we investigate polynomial measure of coherence in this work. First, in Sec.~\ref{sec2}, after briefly reviewing the framework of coherence measure, we define the polynomial coherence measure. Then, in Sec.~\ref{sec3}, we show a no-go theorem for polynomial coherence measures. That is, if the coherence measure just vanishes on incoherent states, there is no such polynomial coherence measure when system dimension is larger than $2$.  Moreover, in Sec.~\ref{sec4}, we permit some superposition states to take zero-coherence, and we find a necessary condition for polynomial coherence measures. In Sec.~\ref{secexample}, we construct a polynomial coherence measure for pure states, which shows similar form with the G-concurrence in entanglement measure. In addition, we derive an analytical result for symmetric states and give a lower bound for general states.  Finally, we conclude in Sec.~\ref{sec6}.

\section{Polynomial coherence measure}\label{sec2}
Let us start with a brief review on the framework of coherence measure \cite{Baumgratz2014Quantify}. In a $d$-dimensional Hilbert space $\mathcal{H}_d$, the coherence measure is defined in a reference basis ${I} =\{\ket{i}\}_{i=1,2,...,d}$. Thus, the incoherent states are the states whose density matrices are diagonal,
\begin{equation}\label{}
\begin{aligned}
\delta=\sum_{i=1}^dp_i\ket{i}\bra{i}.
\end{aligned}
\end{equation}
Denote the set of the incoherent states to be $\mathcal{I}$. The incoherent operation can be expressed as an ICPTP map $\Phi_{ICPTP}(\rho)=\sum_n K_n\rho K_n^\dag$, in which each Kraus operator satisfies the condition $K_n\rho K_n^\dag/Tr(K_n\rho K_n^\dag) \in \mathcal{I}$ if $\rho\in \mathcal{I} $. That is to say, no coherence can be generated from any incoherent states via incoherent operations. Here, the probability to obtain the $n$th output is denoted by $p_n=\mathrm{Tr}(K_n\rho K_n^\dag)$.

Generally speaking, a coherence measure $C(\rho)$ maps a quantum state $\rho$ to a non-negative number. There are three criteria for $C(\rho)$, as listed in Table.~\ref{table_coherence} \cite{Baumgratz2014Quantify}. Note that the criterion $(C1')$ is a stronger version than $(C1)$. Sometimes, a weaker version of $(C2)$ is used, where the monotonicity holds only for the average state, $C(\rho)\geq C(\Phi_{\rm ICPTP}(\rho))$. In this work, we focus on the criterion $(C2)$, since it is more reasonable from the physics point of view.

\begin{table*}[tbph]
\caption{Criteria for a coherence measure}\label{table_coherence}
\centering
\begin{tabular}{ll}
\hline
\hline
$(C1)$ & $C(\delta)=0$ if $\delta\in\rm \mathcal{I}$; $(C1')$ $C(\delta)=0$ iff $\delta\in\rm \mathcal{I}$ \\
$(C2)$ & Monotonicity with post-selection: for any incoherent operation $\Phi_{\rm ICPTP}(\rho)=\sum_n K_n\rho K_n^\dag$,  \\
& $C(\rho)\geq p_n C(\rho_n)$, where $\rho_n=K_n\rho K_n^\dag/p_n$ and $p_n=\mathrm{Tr}(K_n\rho K_n^\dag)$ \\
$(C3)$ & Convexity: $\sum_e p_e C(\rho_e)\geq C(\sum_e p_e \rho_e)$ \\
\hline
\hline
\end{tabular}
\end{table*}

Next, we give the definition of the polynomial coherence measure, drawing on the experience of polynomial invariant for entanglement measure. Denote a homogenous polynomial function of degree-$h$, constructed by the coefficients of a pure state $\ket{\psi}=\sum_{i=1}^d a_i\ket{i}$ in the computational basis, as
\begin{equation}\label{poly:homo}
\begin{aligned}
P_h(\ket{\psi})=\sum_{k_1,k_2,...,k_d}\chi_{k_1k_2\cdots k_d}\prod_{i=1}^{d} a^{k_i}_i,
\end{aligned}
\end{equation}
where $k_i$ are the nonnegative integer power of $a_i$, $\sum{k_i}=h$, and $\chi_{k_1k_2\cdots k_d}$ are coefficients. Then after imposing a proper power $m>0$ on the absolute value of a homogenous polynomial, one can construct a coherence measure as,
\begin{equation}\label{polypure}
\begin{aligned}
C_p(\ket{\psi})=|P_h(\ket{\psi})|^m,
\end{aligned}
\end{equation}
where the overall degree is $hm$, and the subscript $p$ is the abbreviation for polynomial.

A polynomial coherence measure for pure states $C_p(\ket{\psi})$ can be extended to mixed states by utilizing the convex-roof construction,
\begin{equation}\label{convex_roof}
\begin{aligned}
C_p(\rho)=\min\limits_{\{p_i,\ket{\psi_i}\}} \sum_i p_i C_p(\ket{\psi_i}),
\end{aligned}
\end{equation}
where the minimization runs over all the pure state decompositions of $\rho=\sum_i p_i \ket{\psi_i}\bra{\psi_i}$ with $\sum_i p_i=1$ and $p_i\geq 0$, and $C_p(\ket{\psi})$ is the pure-state polynomial coherence measure as shown in Eq.~\eqref{polypure}. Note that if the pure-state measure Eq.~\eqref{polypure} satisfies the coherence measure criteria listed in Table \ref{table_coherence}, the mixed-state measure via the convex-roof construction Eq.~\eqref{convex_roof} would also satisfy these criteria \cite{yuan2015intrinsic}.

%
%
%

\section{No-go theorem}\label{sec3}
The simplest example of the polynomial coherence measure is the $l_1$-norm for $d=2$ on pure state. For a pure qubit state, $\ket{\psi}=\alpha\ket{0}+\beta\ket{1}$, the $l_1$-norm coherence measure takes the sum of the absolute value of the off-diagonal terms in the density matrix,
\begin{equation} \label{poly:l1norm}
\begin{aligned}
C_{l_1}(\ket{\psi})=|\alpha\beta^*|+|\alpha^*\beta|=2|\alpha\beta|.
\end{aligned}
\end{equation}
By the definition of Eq.~\eqref{polypure}, $C_{l_1}$ is the absolute value of a degree-$2$ homogenous polynomial function with a power $m=1$. Meanwhile, this coherence measure satisfies the criteria $(C1')$, $(C2)$, and $(C3)$ \cite{Baumgratz2014Quantify}. Then its convex-roof construction via Eq.~\eqref{poly:l1norm} turns out to be a polynomial coherence measure satisfying these criteria. Note that when the function Eq.~\eqref{poly:l1norm} is extended to $d>2$, it cannot be expressed as the absolute value of a homogenous polynomial function. Thus, when $d>2$, the $l_1$-norm coherence measure is not a polynomial coherence measure.

Surprisingly, for $d>2$, there is no polynomial coherence measure that satisfies the criterion $(C1')$. In order to show this no-go theorem, we first prove the following Lemma.

\begin{lemma}\label{lemma_root}
For any polynomial coherence measure $C_p(\ket{\psi})$ and two orthogonal pure states $\ket{\psi_1}$ and $\ket{\psi_2}$, there exists two complex numbers $\alpha$ and $\beta$ such that
\begin{equation}
\begin{aligned}
C_p(\alpha \ket{\psi_1} + \beta \ket{\psi_2})=0
\end{aligned}
\end{equation}
where $|\alpha |^2+|\beta |^2=1$. That is, there exists at least one zero-coherence state in the superposition of $\ket{\psi_1}$ and $\ket{\psi_2}$.
\end{lemma}

\begin{proof}
Since $m>0$, the roots of $C_p(\ket{\psi})=0$ in Eq.~\eqref{polypure} are the same with the ones of $|P_h(\ket{\psi})|=0$ in Eq.~\eqref{poly:homo}. That is, we only need to prove Lemma for the case of $m=1$. Since $P_h(\ket{\psi})$ is a homogenous polynomial function of the coefficients of $\ket{\psi}$, one can ignore its global phase. Thus, any pure state in the superposition of $\ket{\psi_1}$ and $\ket{\psi_2}$ can be represented by
\begin{equation}
\begin{aligned}
\ket{\psi}=\frac{\ket{\psi_1}+\omega\ket{\psi_2}}{\sqrt{1+|\omega|^2}},
\end{aligned}
\end{equation}
where the global phase is ignored, $\omega$ is a complex number containing the relative phase, and $\ket{\psi}\rightarrow \ket{\psi_2}$, as $|\omega|\rightarrow\infty$.

First, if $C_p(\ket{\psi_2})=0$, the Lemma holds automatically. When $C_p(\ket{\psi_2})>0$, $C_p(\ket{\psi})$ can be written as,
\begin{equation}\label{eq_lemma_root}
\begin{aligned}
C_p(\ket{\psi})&=\left|P_h\left(\frac{\ket{\psi_1}+\omega\ket{\psi_2}}{\sqrt{1+|\omega|^2}}\right)\right|,\\
&=(1+|\omega|^2)^{-h/2}|P_h(\ket{\psi_1}+\omega\ket{\psi_2})|, \\
\end{aligned}
\end{equation}
since $P_h$ a homogenous polynomial function of degree $h$. Note that the condition $C_p(\ket{\psi_2})>0$, i.e.,
\begin{equation}\label{}
\begin{aligned}
\lim_{\omega\rightarrow\infty}(1+|\omega|^2)^{-h/2}|P_h(\ket{\psi_1}+w \ket{\psi_2})|>0,
\end{aligned}
\end{equation}
guarantees that the coefficient of $\omega^h$ in $P_h(\ket{\psi_1}+\omega \ket{\psi_2})=0$ is nonzero. Then, there are $h$ roots of the homogenous polynomial function of $\omega$,
\begin{equation}\label{}
\begin{aligned}
P_h(\ket{\psi_1}+\omega\ket{\psi_2})=0,
\end{aligned}
\end{equation}
denoted by $\{z_1,z_2\ldots, z_h\}$. Thus, $C_p(\ket{\psi})$ can be expressed as
\begin{equation}\label{poly:Cproots}
\begin{aligned}
C_p(\ket{\psi})&= A(1+|\omega|^2)^{-h/2}\prod_{i=1}^h|\omega-z_i|,
\end{aligned}
\end{equation}
where $A>0$ is some constant. In summary, we find at least one $\omega$, $\alpha=(1+|\omega|^2)^{-1/2}$ and $\beta=\omega(1+|\omega|^2)^{-1/2}$, such that $C_p(\ket{\psi})=0$

\end{proof}

\begin{theorem}\label{th_coherence}
There is no polynomial coherence measure in $\mathcal{H}_d$ with $d\geq3$ that satisfies the criterion $(C1')$.
\end{theorem}

\begin{proof}
In the following proof, we focus on the case of $d\geq4$ and leave $d=3$ in Appendix~\ref{d=3}. With $d\ge4$, we can decompose $\mathcal{H}_d$ into two orthogonal subspaces $\mathcal{H}_{d_1}\oplus \mathcal{H}_{d_2}$ in the computational basis, i.e. $\mathcal{H}_1=\{\ket{i}_{i=1,\cdots,d_1}\}$ and $\mathcal{H}_2=\{\ket{i}_{i=d_1+1,\cdots,d}\}$ with the corresponding dimensions $d_1$ and $d_2=d-d_1$ both larger than $2$.

Suppose there exist a polynomial coherence measure $C_p(\ket{\psi})$ such that the criterion $(C1')$ listed in Table \ref{table_coherence} can be satisfied. Then, there are exactly $d$ zero-coherence pure states $\ket{i}$ $(i=1,\cdots,d)$, which form the reference basis. One can pick up two coherent states, $\ket{\psi_1}\in\mathcal{H}_{d_1}$ and $\ket{\psi_2}\in\mathcal{H}_{d_2}$. That is, $C_p(\ket{\psi_1})>0$ and $C_p(\ket{\psi_2})>0$. Since two subspaces $\mathcal{H}_{d_1}$ and $\mathcal{H}_{d_2}$ are orthogonal, any superposition of these two states, $\alpha \ket{\psi_1} + \beta \ket{\psi_2}$ with $|\alpha|^2+|\beta|^2=1$, should not equal to any of the reference basis states, i.e., $\alpha \ket{\psi_1} + \beta \ket{\psi_2}\neq \ket{i}, \forall i=1,...,d$. Thus, due to the criterion $(C1')$, we have
\begin{equation}\label{contrac}
\begin{aligned}
C_p(\alpha \ket{\psi_1} + \beta \ket{\psi_2})>0.
\end{aligned}
\end{equation}
On the other hand, for the polynomial coherence measure $C_p(\ket{\psi})$, Lemma \ref{lemma_root} states that provided any two orthogonal pure states $\ket{\psi_1}$, $\ket{\psi_2}$, there exists at least a pair of complex numbers, $\alpha$ and $\beta$, such that $\alpha \ket{\psi_1} + \beta \ket{\psi_2}$ is a zero-coherence state, i.e.,
\begin{equation}\label{contrac}
\begin{aligned}
C_p(\alpha \ket{\psi_1} + \beta \ket{\psi_2})=0.
\end{aligned}
\end{equation}
Therefore, it leads to a contradiction.

\end{proof}

\section{Necessary condition for polynomial coherence measure}\label{sec4}
From Theorem \ref{th_coherence}, we have shown a no-go result of the polynomial coherence measure for $d\geq 3$ when the criterion $(C1')$ in Table~\ref{table_coherence} is considered. In the following discussions, we study the polynomial coherence measure with the criteria $(C1)$, $(C2)$, and $(C3)$. Then, there will be some coherent states whose coherence measure is zero. This situation also happens in entanglement measures, such as negativity, which remains zero for the bound entangled states \citep{HORODECKI1997333}. Here, we focus on the pure-state case and employ the convex-roof construction for general mixed states. As presented in the following theorem, we find a very restrictive necessary condition for polynomial coherence measures that $C_p(\ket{\psi})=0$, for all \ket{\psi} whose support does not span all the reference basis $\{i\}$.

\begin{theorem}\label{th_coherence_d}
For any $\ket{\psi}\in \mathcal{H}_d$, the value of a polynomial coherence measure $C_p(\ket{\psi})$ should vanish if the rank of the corresponding dephased state $\Delta(\ket{\psi}\bra{\psi})$ is less than $d$, i.e., $rank(\Delta(\ket{\psi}\bra{\psi}))<d$.
\end{theorem}

\begin{proof}
Suppose there exists $\ket{\psi_1}\in \mathcal{H}_d$ such that $C_p(\ket{\psi_1})>0$ and  $rank(\Delta(\ket{\psi_1}\bra{\psi_1}))=d_1<d$. Without loss of generality, $\ket{\psi_1}$ is assumed to be in the subspace $\mathcal{H}_{d_1}=spanned\{\ket{1},\ket{2},...,\ket{d_1}\}$. Define the complementary subspace to be $\mathcal{H}_{d_2}=spanned\{\ket{d_1+1},\ket{d_1+2},...,\ket{d}\}$, where $d_2=d-d_1>0$.

\textbf{Step 1, we show that if $d_1\leq d/2$, then $C_p(\ket{\psi_1})>0$ leads to a contradiction to Lemma \ref{lemma_root}.} Now that $d_1\leq d/2 \leq d_2$, there exists a relabeling unitary $U_t$ that transforms the bases in $\mathcal{H}_{d_1}$ to parts of the bases in $\mathcal{H}_{d_2}$. For instance, $\mathcal{H}_{d_1}=spanned\{\ket{1}, \ket{2}\}$ and $\mathcal{H}_{d_2}=spanned\{\ket{3}, \ket{4}, \ket{5}\}$, $U_t$ can be chosen as $\ket{1}\bra{3}+\ket{3}\bra{1}+\ket{2}\bra{4}+\ket{4}\bra{2}$. In fact, $U_t$ and $U_t^\dag$ are both incoherent operation, since they just exchange the index of the reference basis. Assume that $U_t$ maps $\ket{\psi_1}$ to a new state $\ket{\psi_2}=U_t\ket{\psi_1} \in \mathcal{H}_{d_2}$, then we have $\braket{\psi_1}{\psi_2}=0$. Due to the criterion $(C2)$, it is not hard to show that an incoherent unitary transformation does not change the coherence,
\begin{equation}\label{equal_coherence}
\begin{aligned}
C_p(\ket{\psi_1})= C_p(\ket{\psi_2}).
\end{aligned}
\end{equation}

Define another incoherent operation, composed by two operators $P_1=\sum_{i=1}^{d_1}\ket{i}\bra{i}$ and $P_2=\sum_{i=d_1+1}^d\ket{i}\bra{i}$ that project states to $\mathcal{H}_{d_1}$ and $\mathcal{H}_{d_2}$, respectively,
\begin{equation}\label{ICPTP}
\begin{aligned}
\Phi_{ICPTP}(\rho)&=\sum_{i=1,2}P_i\rho P_i^\dag, \\
\end{aligned}
\end{equation}
which represents a dephasing operation between the two subspaces. Then, for any superposition state, $\alpha \ket{\psi_1} + \beta \ket{\psi_2}$ with $|\alpha|^2+|\beta|^2=1$, its coherence measure should not increase under the ICPTP operation, as  required by $(C2)$ in Table.~\ref{table_coherence},
\begin{equation}\label{}
\begin{aligned}
C_p(\alpha \ket{\psi_1} + \beta \ket{\psi_2}) &\ge C_p(\Phi_{ICPTP}(\alpha \ket{\psi_1} + \beta \ket{\psi_2})) \\
&= |\alpha|^2 C_p(\ket{\psi_1})+|\beta|^2 C_p(\ket{\psi_2})\\
&=C_p(\ket{\psi_1})>0.
\end{aligned}
\end{equation}
where the last equality comes from Eq.~\eqref{equal_coherence}. Therefore, $C_p(\alpha \ket{\psi_1} + \beta \ket{\psi_2})>0$ for any $\alpha$ and $\beta$. This leads to a contradiction to Lemma \ref{lemma_root}.

\textbf{Step 2, we show that if $d/2<d_1<d$, then $C_p(\ket{\psi_1})>0$ also leads to a contradiction.} Now that $0<d_2< d/2 < d_1<d$, for any $\ket{\psi_2}\in\mathcal{H}_{d_2}$, we have $C_p(\ket{\psi_2})=0$ due to the above proof in Step 1.

Similar to the proof of Lemma \ref{lemma_root}, we only need to consider the case of $m=1$ and we can get the coherence measure for the superposition state of $\ket{\psi_1}\in\mathcal{H}_{d_1}$ and $\ket{\psi_2}\in\mathcal{H}_{d_2}$ as $(1+|\omega|^2)^{-h/2}|P_h(\ket{\psi_1}+\omega\ket{\psi_2})|$. Since
\begin{equation}\label{}
\begin{aligned}
C_p(\ket{\psi_2})=\lim_{\omega\rightarrow\infty}(1+|\omega|^2)^{-h/2}|P_h(\ket{\psi_1}+w \ket{\psi_2})|=0,
\end{aligned}
\end{equation}
the largest degree of $\omega$ in the polynomial $P_h(\ket{\psi_1}+\omega\ket{\psi_2})$, denoted by $\mu$, is smaller than the degree $h$.

When $\mu=0$, i.e., the polynomial is a constant, we denote its absolute value by $k$. Then the coherence measure becomes,
\begin{equation}\label{k_decrease}
\begin{aligned}
C_p(\ket{\psi})=k(1+|\omega|^2)^{-h/2}.
\end{aligned}
\end{equation}
We show that the constant $k=0$ in Appendix \ref{k=0}. As a result, $C_p(\ket{\psi_1})=0$. This leads to a contradiction to our assumption that $C_p(\ket{\psi_1})>0$.

When $0<\mu<d$, i.e., $P_h(\ket{\psi_1}+\omega\ket{\psi_2})$ is a non-constant polynomial of $\omega$, there exists at least one root $|z|<\infty$, such that $P_h(\ket{\psi_1}+z\ket{\psi_2})=0$. Then, we can find that the coherence measure of the state $\ket{\psi_r}=(\ket{\psi_1}+z\ket{\psi_2})/\sqrt{1+|z|^2}$ is $C_p(\ket{\psi_r})=0$. Next, we apply the ICPTP operation described in Eq.~\eqref{ICPTP} on $\ket{\psi_r}$ and obtain,
\begin{equation}\label{neq_coherence}
\begin{aligned}
C_p(\ket{\psi_r})&\geq \frac{1}{1+|z|^2}C_p(\ket{\psi_1})+\frac{|z|^2}{1+|z|^2} C_p(\ket{\psi_2})\\
&=\frac{1}{1+|z|^2}C_p(\ket{\psi_1}),
\end{aligned}
\end{equation}
where we use $C_p(\ket{\psi_2})=0$ in the equality. Combing the fact that $C_p(\ket{\psi_r})=0$, we can reach the conclusion that $C_p(\ket{\psi_1})=0$. This leads to a contradiction to our assumption that $C_p(\ket{\psi_1})>0$.

\end{proof}

\section{G-coherence measure} \label{secexample}
From Theorem \ref{th_coherence_d}, we can see that only the states with a full support on the computational basis could have positive values of a polynomial coherence measure. Here, we give an example of polynomial coherence measure satisfying this condition, which takes the geometric mean of the coefficients, for $\ket{\psi}=\sum_{i=1}^d a_i\ket{i}$,
\begin{equation}\label{G_measure}
\begin{aligned}
C_G(\ket{\psi})=d|a_1a_2...a_d|^{2/d}.
\end{aligned}
\end{equation}
Note that it is a degree-$d$ homogenous polynomial function modulated by a power $m=2/d$. This definition is an analogue to the G-concurrence in entanglement measure, which is related to the geometric mean of the Schmidt coefficients of a bipartite pure state \cite{gour2005family}. Hence we call the coherence measure defined in Eq.~\eqref{G_measure} \emph{G-coherence measure}. Since the geometric mean function is a concave function \cite{boyd2004convex}, following Theorem 1 in Ref.~\cite{Du2015CoherenceM}, we can quickly show that the G-coherence measure satisfies the criteria $(C1)$, $(C2)$ and $(C3)$.

When $d=2$, the G-coherence measure becomes the $l_1$-norm measure on pure state. When $d>2$, according to Theorem \ref{th_coherence_d}, there is a significant amount of coherent states whose $C_G$ is zero. For instance, in the case of $d=3$, the state $\frac{1}{\sqrt{2}}(\ket{0}+\ket{1})$ has zero G-coherence and this state cannot be transformed to a coherent state $\ket{\psi}$, where $rank(\Delta(\ket{\psi}\bra{\psi}))=3$, via a probabilistic incoherent operation \cite{winter2016operational}.

Now we move onto the mixed states with the convex-roof construction. In fact, searching for the optimal decomposition in Eq.~\eqref{convex_roof} is generally hard. However, like the entanglement measures, there exist analytical solutions for the states with symmetry  \cite{terhal2000entanglement,vollbrecht2001entanglement}. Here, we study the states related to the permutation group $G_s$ on the reference basis. A element $g\in G_s$ is defined as
\begin{equation}\label{}
\begin{aligned}
g=\left(
  \begin{array}{cccc}
    1 & 2 & ...&d\\
    i_1 & i_2&...&i_d\\
  \end{array}
\right)
\end{aligned}
\end{equation}
and the order (the number of the elements) of $G_s$ is $d!$. The corresponding unitary of $g$ is denoted as $U_g=\sum_k \ket{i_k}\bra{k}$. Then we have the following definition.
\begin{definition}\label{symmetry_state}
A state $\rho$ is a symmetric state if it is invariant under all the permutation unitary operations, i.e., $\forall g\in G_s$, $U_g\rho U_g^\dag=\rho$.
\end{definition}
Denote the symmetric state as $\rho^s$ and the symmetric state set as $S$. Given the maximally coherent state $\ket{\Psi_d}=\frac1{\sqrt{d}}\sum_{i}\ket{i}$, it is not hard to show the explicit form of symmetric states,
\begin{equation}\label{poly:symmstate}
\rho^s=p\ket{\Psi_d}\bra{\Psi_d}+(1-p)\frac{\mathbb{I}}{d},
\end{equation}
which is only determined by a single parameter, the mixing probability $p\in[0,1]$. Apparently, the symmetric state $\rho^s$ is a mixture of the maximally coherent state $\ket{\Psi_d}$ and the maximally mixed state $\mathbb{I}/d$. The state $\ket{\Psi_d}$ is the only pure state in set $S$. Borrowing the techniques used in quantifying entanglement of symmetric states \cite{vollbrecht2001entanglement,sentis2016quantifying}, we obtain an analytical result $C_G(\rho^s)$ in Theorem \ref{Th:CGsymmetric}, following Lemma \ref{twoproperty} and Lemma \ref{convex_setS}.

First, we consider a map
\begin{equation}\label{mixingmap}
\begin{aligned}
\Lambda(\rho)=\frac1{|G_s|}\sum_g U_g \rho U_g^\dag.
\end{aligned}
\end{equation}
It uniformly mixes all the permutation unitary $U_g$ on a state $\rho$, which is an incoherent operation by definition.

\begin{lemma}\label{twoproperty}
The map $\Lambda(\rho)$ defined in Eq.~\eqref{mixingmap} satisfies two properties, $\forall \rho$,
\begin{enumerate} [(1)]
\item
$\Lambda(\rho)\in S$, i.e., the output state is a symmetric state, as defined in Definition \ref{symmetry_state};
\item
$\bra{\Psi_d}\rho \ket{\Psi_d}=\bra{\Psi_d}\Lambda(\rho)\ket{\Psi_d}$, i.e., the map $\Lambda(\rho)$ does not change the overlap with the maximally coherent state $\ket{\Psi_d}$.
\end{enumerate}
\end{lemma}

\begin{proof}
For any $U_{g'}$ with $g'\in G_s$,
\begin{equation}\label{}
\begin{aligned}
U_{g'} \Lambda(\rho)U_{g'}^\dag
&=\frac{1}{|G_s|}\sum_g (U_{g'}U_g) \rho (U_{g'}U_g)^\dag\\
&=\frac{1}{|G_s|}\sum_g U_{g'g} \rho U_{g'g}^\dag\\
&=\Lambda(\rho).
\end{aligned}
\end{equation}
The last equality is due to the fact that by going through all permutations $g$, the joint permutation $g'g$ also traverses all the permutations in the group $G_s$. By Definition \ref{symmetry_state}, we prove that $\Lambda(\rho)\in S$.

The overlap between the output state $\Lambda(\rho)$ and the maximally coherent state $\ket{\Psi_d}$ is given by,
\begin{equation}\label{}
\begin{aligned}
\bra{\Psi_d}\Lambda(\rho)\ket{\Psi_d}&=\bra{\Psi_d}\frac{1}{|G_s|}\sum_g U_g \rho U_g^\dag\ket{\Psi_d}\\
&=\frac{1}{|G_s|}\sum_g\bra{\Psi_d}U_{g^{-1}}^\dag\rho U_{g^{-1}}\ket{\Psi_d}\\
&=\bra{\Psi_d}\rho\ket{\Psi_d}.
\end{aligned}
\end{equation}
where in the second line we use the relation $U_g^\dag=U_{g^{-1}}$ and the last line is due to the fact that $\ket{\Psi_d}\in S$ and $U_{g^{-1}}\ket{\Psi_d}=\ket{\Psi_d}$.
\end{proof}

Then, we define the following function for a symmetric state $\rho^s$,
\begin{equation}\label{convexroof_epsilon}
\bar{C}_G(\rho^s)=\min_{\ket{\psi}}\{C_G(\ket{\psi})|\Lambda(\ket{\psi}\bra{\psi})=\rho^s\}.
\end{equation}
Since the state $\rho^s$ in Eq.~\eqref{poly:symmstate} only has one parameter $p$, it can be uniquely determined by its overlap with the maximally coherent state $K=\bra{\Psi_d}\rho^s\ket{\Psi_d}=p\frac{d-1}{d}+\frac{1}{d}$. Thus, $\rho^s$ linearly depends on $K$. According to Lemma \ref{twoproperty}, $\Lambda(\ket{\psi}\bra{\psi})$ is a symmetric state and the overlap does not change under the map $\Lambda$. Hence, the constraint $\Lambda(\ket{\psi}\bra{\psi})=\rho^s$ in Eq.~\eqref{convexroof_epsilon} is equivalent to $|\bra{\Psi_d}\psi\rangle|^2=\bra{\Psi_d}\rho^s\ket{\Psi_d}$. Following the derivations of the G-concurrence \cite{sentis2016quantifying}, we solve the minimization problem and obtain an explicit form of $\bar{C}_G(\rho^s)$,
\begin{equation}\label{purestate_optimal}
\bar{C}_G(K)=\left\{
\begin{aligned}
&0 &0\leq K\leq \frac{d-1}{d}\\
&d(ab^{d-1})^{(2/d)} &\frac{d-1}{d}\leq K \leq 1
\end{aligned}
\right.
\end{equation}
where
\begin{equation*}\label{}
\begin{aligned}
&a=\frac{1}{\sqrt{d}}(\sqrt{K}-\sqrt{d-1}\sqrt{1-K}),\\
&b=\frac{1}{\sqrt{d}}(\sqrt{K}+\frac{\sqrt{1-K}}{\sqrt{d-1}}).
\end{aligned}
\end{equation*}
Details can be found in Appendix \ref{App_purestate_optimal}. Here, we substitute $\bar{C}_G(K)$ for $\bar{C}_G(\rho^s)$ without ambiguity. When $\frac{d-1}{d}\leq K \leq 1$, $\bar{C}_G(K)$ is a concave function \cite{sentis2016quantifying}. We show $\bar{C}_G(K)$ in the case of $d=4$ in Fig.~\ref{Fig_concave}. Moreover, following the results of Ref.~\cite{vollbrecht2001entanglement}, we have the following lemma.

\begin{lemma}\label{convex_setS}
The convex-roof of the G-coherence measure $C_G$ for a symmetric state $\rho^s$ is given by,
\begin{equation}\label{Eq:convex_setS}
\begin{aligned}
C_G(\rho^s)&=\min_{\{p_i,\ket{\psi_i}\}}\sum_i p_i C_G(\ket{\psi_i})\\
            &= \min_{\{q_j,\rho_j^s\}} \sum_j q_j\bar{C}_G(\rho_j^s),
\end{aligned}
\end{equation}
where $\sum_i p_i \ket{\psi_i}\bra{\psi_i}=\rho^s$, $\sum_j q_j \rho_j^s=\rho^s$, and $\rho_j^s\in S$.
\end{lemma}
\begin{proof}
Denote $Z_1=\min_{\{p_i,\ket{\psi_i}\}}\sum_i p_i C_G(\ket{\psi_i})$ and $Z_2=\min_{\{q_j,\rho_j^s\}} \sum_j q_j\bar{C}_G(\rho_j^s)$. Now we prove the lemma by showing that both of them equal to,
\begin{equation}\label{middleCG}
\begin{aligned}
Z_3=\min_{\{p_i,\ket{\psi_i}\}}\left\{\sum_i p_i C_G(\ket{\psi_i})\bigg|\sum_ip_i\Lambda(\ket{\psi_i}\bra{\psi_i})=\rho^s\right\}.
\end{aligned}
\end{equation}

$Z_1=Z_3$: For a decomposition, $\rho^s=\sum_i p_i \ket{\psi_i}\bra{\psi_i}$, after applying the map $\Lambda$ on both sides, we have
\begin{equation}\label{}
\begin{aligned}
\sum_ip_i\Lambda(\ket{\psi_i}\bra{\psi_i})=\Lambda(\rho^s)=\rho^s.
\end{aligned}
\end{equation}
Here, we use the fact that $\rho^s$ is a symmetric state, which is invariant under the map $\Lambda$. That is, any decomposition satisfies the constraint $\sum_i p_i \ket{\psi_i}\bra{\psi_i}=\rho^s$ as required for $Z_1$ also satisfies the constraint $\sum_ip_i\Lambda(\ket{\psi_i}\bra{\psi_i})=\rho^s$ as required for $Z_3$. Thus, we have $Z_3\leq Z_1$. On the other hand, the constraint $\sum_ip_i\Lambda(\ket{\psi_i}\bra{\psi_i})=\rho^s$  in Eq.~\eqref{middleCG} is also a pure-state decomposition of the state $\rho^s$, since every component in $\Lambda(\ket{\psi_i}\bra{\psi_i})$ is a pure state $U_g
\ket{\psi_i}$ with probability $p_i/|G_s|$. Thus we also have $Z1\leq Z3$. Consequently, $Z1=Z3$.

$Z_2=Z_3$: In fact, the constraint in Eq.~\eqref{middleCG} is on $\Lambda(\ket{\psi_i}\bra{\psi_i})\in S$, thus we can solve the minimization problem of Eq.~\eqref{middleCG} in two steps. First, given $\Lambda(\ket{\psi_i}\bra{\psi_i})\in S$, we minimize $C_G(\ket{\psi_i})$, which turns out to be the same as the definition of $\bar{C}_G(\Lambda(\ket{\psi_i}\bra{\psi_i}))$ in Eq.~\eqref{convexroof_epsilon}. Next, we optimize the decomposition of $\rho^s$ in the symmetric state set $S$, which turns out to be the same as the definition of $Z_2$. Thus we have $Z_2=Z_3$.
\end{proof}

\begin{figure}[thb]
\centering
\resizebox{7.5cm}{!}{\includegraphics[scale=0.8]{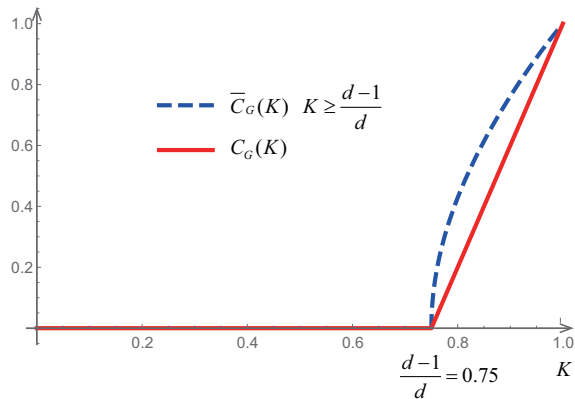}}
\caption{Illustration for the two functions $\bar{C}_G(K)$ and $C_G(K)$ in $d=4$ case. When $0\leq K\leq \frac{d-1}{d}=0.75$, $\bar{C}_G(K)=0$; when $\frac{d-1}{d}=0.75\leq K\leq1$,  $\bar{C}_G(K)$ is a concave function following the form in Eq.~\eqref{purestate_optimal}, represented by the dashed blue line. Thus the minimization result via Eq.~\eqref{convex:on:K},  $C_G(K)$ is the linear function $1-4(1-K)$, when $\frac{d-1}{d}=0.75\leq K\leq1$, described by the red line.}\label{Fig_concave}
\end{figure}

\begin{theorem}\label{Th:CGsymmetric}
For a symmetric state $\rho^s \in S$ in $\mathcal{H}_d$, the G-coherence measure is given by
\begin{equation}\label{CGsymmetric}
\begin{aligned}
C_G(\rho^s) &=\max \{1-d(1-K),0\}, \\
\end{aligned}
\end{equation}
where $K=\bra{\Psi_d}\rho^s\ket{\Psi_d}$ is the overlap between $\rho_s$ and the maximally coherent state $\ket{\Psi_d}$.
\end{theorem}
\begin{proof}
According to Lemma \ref{convex_setS}, the G-coherence measure for a symmetric state is given by $C_G(\rho^s)=\min_{\{q_j,\rho_j^s\}} \sum_j q_j\bar{C}_G(\rho_j^s)$ with $\sum_j q_j \rho_j^s=\rho^s$. Since the symmetric state linearly depends on the overlap $K$, this minimization is equivalent to,
\begin{equation}\label{convex:on:K}
\begin{aligned}
C_G(K)=\min_{\{q_j,K_j\}} \left\{ \sum_j q_j\bar{C}_G(K_j)\bigg|\sum_j q_jK_j=K\right\}\\
\end{aligned}
\end{equation}
Then, according to the explicit expression of $\bar{C}_G(K)$ in Eq.~\eqref{purestate_optimal}: When $0\leq K\leq \frac{d-1}{d}$, $\bar{C}_G(K)=0$. Thus, $C_G(K)\leq \bar{C}_G(K)=0$. When $\frac{d-1}{d}\leq K \leq 1$, fortunately, $\bar{C}_G(K)$ is a concave function. It is not hard to find that the optimization result is a straight line connecting the point $\{\frac{d-1}{d},0\}$ and $\{1,1\}$ on the $\{K,C_G(K)\}$ plane. Consequently, $C_G(\rho^s)$ shows the form in Eq.~\eqref{CGsymmetric}.
\end{proof}

The dependence of $\bar{C}_G(K)$ and $C_G(K)$ on $K$ in the case of $d=4$ are plotted in Fig.~\ref{Fig_concave}. Furthermore, we can give a lower bound of the G-coherence measure $C_G$ for any general mixed state $\rho$, with the analytical solution for $\rho^s$ in Theorem \ref{Th:CGsymmetric}.

\begin{corollary}
For a mixed state $\rho$,
\begin{equation}\label{C_Ganystate}
C_G(\rho)\ge \max [1-d(1-K),0]
\end{equation}
where $K=\bra{\Psi_d}\rho\ket{\Psi_d}$.
\end{corollary}

\begin{proof}
Since $\Lambda$ is an incoherent operation, we have,
\begin{equation}\label{}
\begin{aligned}
C_G(\rho)\geq C_G(\Lambda(\rho)).
\end{aligned}
\end{equation}
From Lemma \ref{twoproperty}, we know that the overlap $K=\bra{\Psi_d}\rho\ket{\Psi_d}=\bra{\Psi_d}\Lambda(\rho)\ket{\Psi_d}$ and $\Lambda(\rho)\in S$. Following Theorem \ref{Th:CGsymmetric}, the corollary holds.
\end{proof}

In fact, the tightness of the bound depends on the overlap. Thus, we can enhance the bound by pre-treating the state by a certain ICPTP $\chi$ that can increase the overlap, i.e.,
\begin{equation}\label{}
\begin{aligned}
C_G(\rho)\geq C_G(\chi(\rho))\geq C_G(\Lambda(\chi(\rho)))\geq \max [1-d(1-K'),0],
\end{aligned}
\end{equation}
where $K'=\bra{\Psi_d}\chi(\rho)\ket{\Psi_d}>K=\bra{\Psi_d}\rho\ket{\Psi_d}$.

\section{conclusion and outlook}\label{sec6}
In this paper, we give the definition of polynomial coherence measure $C_p(\rho)$, which is an analog to the definition of polynomial invariant in classifying and quantifying the entanglement resource. First, we show that there is no polynomial coherence measure satisfying criterion $(C1')$ in Table.~\ref{table_coherence}, when the dimension of the Hilbert space $d$ is larger than $2$. That is, there always exist some pure states $\ket{\psi}\neq \ket{i}(i=1,...,d)$ possessing zero-coherence when $d\geq3$. Then, we find a very restrictive necessary condition for polynomial coherence measures --- the coherence measure should vanish if the rank of the corresponding dephased state $\Delta(\ket{\psi}\bra{\psi})$ is smaller than the Hilbert space dimension $d$. Meanwhile, we give an example of polynomial coherence measure $C_G(\rho)$, called G-coherence measure. We derive an analytical formula of the convex-roof of $C_G$ for symmetric states, and also give a lower bound of $C_G$ for general mixed state. In addition, we should remark that the symmetry consideration in our paper is also helpful to understand and bound the other coherence measures, especially the ones built by the convex-roof method.

In entanglement quantification, the polynomial invariant is an entanglement monotone if and only if its degree $\eta\leq4$ in the multi-qubit system \cite{verstraete2003normal,eltschka2012multipartite}. Here, the quantification theory of coherence shows many similarities to the one for entanglement. Following the similar approaches in our paper, some results can be extended to the entanglement case. For example, one can obtain some necessary conditions where a polynomial invariant serves as an entanglement monotone, in more general multi-partite system $\mathcal{H}={\mathcal{H}_{d_l}}^{\otimes{N}}$, whose local dimension $d_l>2$ \cite{gour2013classification}.

After finishing the manuscript, we find that a coherence measure similar to $C_G(\rho)$ is also put forward in Ref.~\cite{chin2017generalized}, dubbed generalized coherence concurrence, by analog to the generalized concurrence for entanglement \cite{gour2005family}. However, the analytical solutions and its relationship with polynomial coherence measure are not presented in Ref.~\cite{chin2017generalized}.

\acknowledgments
We acknowledge J.~Ma and T.~Peng for the insightful discussions. This work was supported by the National Natural Science Foundation of China Grants No.~11674193.

\appendix
\section{Proof of Theorem~\ref{th_coherence} for $d=3$}\label{d=3}
In the main part, Theorem~\ref{th_coherence} for the case of $d\geq4$ has been proved. Here we prove the $d=3$ case. First, a Lemma that is an extension of Lemma.\ref{lemma_root} follows.
\begin{lemma}\label{lemma_root_2}
For any polynomial coherence measure $C_p(\ket{\psi})$,
and any two pure quantum states $\ket{\psi_1}, \ket{\psi_2}$ satisfying $|\bra{\psi_2}\psi_1\rangle|<1$, there is at least one zero-coherence state in the superposition space of them.
\end{lemma}
\begin{proof}
Like in Lemma.\ref{lemma_root}, without loss of generality, we just need to consider the scenario of power $m=1$.
First, if $C_p(\ket{\psi_2})=0$, the Lemma holds automatically. So we focus on the $C_p(\ket{\psi_2})\neq0$ case in the following.

Let us denote $\bra{\psi_1}\psi_2\rangle=k e^{i\theta}$ with $k<1$. Then, after ignoring the global phase, any superposition state of $\ket{\psi_1}$ and $\ket{\psi_2}$ can be represented by
\begin{equation}\label{}
\begin{aligned}
\ket{\psi}=\frac{\ket{\psi_1}+\omega\ket{\psi_2}}{Z(\omega)},
\end{aligned}
\end{equation}
where $\omega$ is a complex number and the normalization factor $Z(\omega)=|\ket{\psi_1}+\omega\ket{\psi_2}|=\sqrt{1+|\omega|^2+2|\omega|k\cos(\theta+\theta'})$ with $\omega=|\omega|e^{i\theta'}$.

Similar to Lemma.~\ref{lemma_root}, we can factorize $C_p(\ket{\psi})$ as
\begin{equation}\label{}
\begin{aligned}
C_p(\ket{\psi})&=\left|P_h\left(\frac{\ket{\psi_1}+\omega\ket{\psi_2}}{Z(\omega)}\right)\right|\\
&=\frac{1}{Z(\omega)^h} |P_h(\ket{\psi_1}+\omega\ket{\psi_2})|\\
&=\frac{A'}{Z(\omega)^h}\Pi_{i=1}^h|\omega-z_i|,
\end{aligned}
\end{equation}
where $A'$ is some constant and $z_i(i=1,2,...,h)$ are the roots for the polynomial function $P_h(\ket{\psi_1}+\omega\ket{\psi_2})$. Thus we can find at least one root in this $C_p(\ket{\psi_2})\neq0$ case, or equivalently, a zero-coherence state.
\end{proof}
With the help of Lemma.\ref{lemma_root_2}, now we prove Th.~\ref{th_coherence} for $d=3$ case .
First, similar to the main part, we can choose two states with non-zero coherence as,
\begin{equation}\label{}
\begin{aligned}
\ket{\psi_1}&=\frac{1}{\sqrt{2}}(\ket{1}+\ket{2}),\\ \ket{\psi_2}&=\frac{1}{\sqrt{2}}(\ket{2}+\ket{3}).\\
\end{aligned}
\end{equation}

Even though these two states share overlap with each other, any superposition state
$\alpha \ket{\psi_1} + \beta \ket{\psi_2}$ should not equal to
the pure state $\ket{i}(i=1,2...d)$ in the computational basis.
As required by the criterion $(C1')$ in Table.~\ref{table_coherence}, $\ket{i}(i=1,2...d)$ are the only zero-coherence pure state. Thus,
$C(\alpha \ket{\psi_1} + \beta \ket{\psi_2})>0$.
Nonetheless, it is contradict to Lemma.~\ref{lemma_root_2}. Consequently, there is no polynomial coherence measure satisfying criterion $(C1')$ for $d=3$ case.

\section{Proof for $k=0$ in Eq.~\eqref{k_decrease}}\label{k=0}
In the main part, the coherence measure for the superposition state of $\ket{\psi_1}\in \mathcal{H}_{d_1}$ and $\ket{\psi_2}\in \mathcal{H}_{d_2}$ shows,
\begin{equation*}
\begin{aligned}
C_p(\ket{\psi})=k(1+|\omega|^2)^{-h/2}.
\end{aligned}
\end{equation*}

If $k>0$, the coherence measure strictly decreases with the increasing of $|\omega|$. That is, for any superposition state $\ket{\psi}=(\ket{\psi_1}+\omega\ket{\psi_2})/\sqrt{1+|\omega|^2}$ with $|\omega|>0$, we have $C_p(\ket{\psi})<C_p(\ket{\psi_1})$. We denote the state coefficients by $\alpha=(1+|\omega|^2)^{-1/2}$ and $\beta=\omega(1+|\omega|^2)^{-1/2}$ here. In the following, we show that there exists a state $\ket{\psi}=\alpha \ket{\psi_1}+\beta\ket{\psi_2}$ with $\alpha<1$ (or equivalently $|\omega|>0$), such that $C_p(\ket{\psi})\geq C_p(\ket{\psi_1})$. As a result, this contradiction leads to $k=0$.

From Refs.~\cite{du2015conditions,winter2016operational}, we know that $\ket{\Psi}=\sum_{i=1}^d\Psi_i\ket{i}$ can transform to $\ket{\Phi}=\sum_{i=1}^d\Phi_i\ket{i}$ via incoherent operation, if $(|\Psi_1|^2,...,|\Psi_d|^2)^t$ is majorized by $(|\Phi_1|^2,...,|\Phi_d|^2)^t$.
Then combing the criteria $(C2)$ and $(C3)$ in Table.~\ref{table_coherence}, we obtain that the coherence measure is non-increasing after incoherent operation. Thus, $C(\ket{\Psi})\geq C(\ket{\Phi})$ for any coherence measure.

In our case,
first, we denote $\ket{\psi_1}=\sum_{i=1}^{d_1}a_i\ket{i}$ with $\forall i, |a_i|>0$. And choose
$\ket{\psi_2}=\frac{1}{\sqrt{d_2}}\sum_{i=d_1+1}^d\ket{i}$.
Then we can build a state $\ket{\psi}=\alpha \ket{\psi_1} + \beta \ket{\psi_2}$ that satisfies $\alpha<1$ and $C_p(\ket{\psi})\geq C_p(\ket{\psi_1})$, with the help of the aforementioned majorization condition.

To be specific, if $\alpha$ satisfying,
\begin{equation}\label{majorization}
\begin{aligned}
\alpha^2|a_j|^2\geq \beta^2/d_2,
\end{aligned}
\end{equation}
where $|a_j|^2$ is the minimal value in $\{|a_i|^2\}$, then
$(\alpha^2|a_1|^2,\alpha^2|a_2|^2,...,\alpha^2|a_{d_1}|^2,\beta^2/d_2,...,\beta^2/d_2)^t$ is majorized by $(|a_1|^2,|a_2|^2,...,|a_{d_1}|^2,0,...,0)^t$. Thus, $C_p(\ket{\psi})\geq C_p(\ket{\psi_1})$.
In fact, $\alpha=(d_2|a_j|^2+1)^{-1/2}<1$, when the inequality is saturated in Eq.~\ref{majorization}.
\section{Derivation of Eq.~\eqref{purestate_optimal}} \label{App_purestate_optimal}
As pointed in the main part, the constraint for the pure state $\ket{\psi}=\sum_i a_i\ket{i}$ in Eq.~\eqref{convexroof_epsilon} is the overlap $K=|\bra{\Psi_d}\psi\rangle|^2$, i.e.,
\begin{equation}\label{c1}
\begin{aligned}
|\sum_i a_i|=\sqrt{dK},
\end{aligned}
\end{equation}
and the coefficients $a_i$ of the state should also satisfy the normalization condition,
\begin{equation}\label{c2}
\sum_i |a_i|^2=1.
\end{equation}

When $0\leq K\leq \frac{d-1}{d}$, we can always set one of the coefficients $a_j=0$ with $j\in\{i\}$, and let the corresponding $C_G$ equal to $0$. Thus $\bar{C}_G(K)=0$ in this $K$ domain.

On the other hand, all the coefficients $a_i\neq 0$, when $\frac{d-1}{d}\leq K \leq 1$. In this $K$ domain, we should minimize $C_G(\ket{\psi})=d(\Pi_i|a_i|)^{\frac{2}{d}}$ under the constraints in Eq.~\eqref{c1} and Eq.~\eqref{c2}. Note that $\sum_i|a_i|\geq|\sum_ia_i|$ and the equality can be reached when the coefficients share the same phase. Thus the constraint in Eq.~\eqref{c1} can be replaced by,
\begin{equation}\label{c3}
\begin{aligned}
\sum_i |a_i|=\sqrt{dK}.
\end{aligned}
\end{equation}

In fact, the function optimized here is the same to the one in Ref.~\cite{sentis2016quantifying} for the G-concurrence, after substituting the Schmidt coefficients for the state coefficients $|a_i|$.  Thus, utilizing the same Lagrange multipliers in Supplemental Material of Ref.~\cite{sentis2016quantifying}, we can obtain Eq.~\eqref{purestate_optimal} in the main part. And we can show that $\bar{C}_G(K)$ is a concave function, when $\frac{d-1}{d}\leq K \leq 1$, by directly following the derivation there.
\bibliographystyle{apsrev4-1}

\bibliography{bibpolycoherence}

\end{document}